\documentclass{article}
\pdfoutput=1 
\usepackage{graphicx,amsthm,amssymb,amsmath,url}
\usepackage[thinlines]{easybmat}
\usepackage[letterpaper,hmargin=1.3in,vmargin=1.2in]{geometry}

\theoremstyle{plain} 
\newtheorem{thm}{Theorem}[section]
\newtheorem{lem}{Lemma}[section]

\newtheorem{cor}{Corollary}[section]
\newtheorem{defn}{Definition}
\newtheorem{prop}{Proposition}

\newcommand{\p}{\mathbf{p}}
\newcommand{\ot}{\tilde{\omega}}
\newcommand{\R}{\mathbb{R}}
\newcommand{\etal}{\emph{et~al.}}

\title{Small grid embeddings of 3-polytopes}

\author{  Ares {Rib\'o Mor}\thanks{Gesellschaft zur F{\"o}rderung angewandter Informatik e.V., Berlin, Germany,  \texttt{ribo@gfai.de}.
Partially supported by the
Deutsche Forschungsgemeinschaft within the European Research Training Network \emph{Combinatorics,
Geometry and Computation} (No.~GRK~588/2).}
	\and
	G\"unter Rote\thanks{Institut f\"ur Informatik, Freie Universit\"at Berlin, Germany, \texttt{rote}\texttt{@inf.fu-berlin.de}.}
	\and
	Andr\'e Schulz\thanks{Institut f\"ur Mathematsche Logik und Grundlagenforschung, Universit\"at M\"unster, \protect\url{andre.schulz@uni-muenster.de}. Partially supported by
the German Research Foundation (DFG) under grant SCHU 2458/1-1.}
}

\begin{document}

\maketitle

\begin{abstract}
We introduce an algorithm that embeds a given 3-connected planar graph as a convex 3-polytope with integer coordinates. The size of the  coordinates is bounded by $O(2^{7.55n})=O(188^{n})$. If the graph contains a triangle we can bound the integer coordinates by $O(2^{4.82n})$. If the graph contains a quadrilateral we can bound the integer coordinates by $O(2^{5.46n})$. The crucial part of the algorithm is to find a convex plane embedding whose edges can be weighted such that the sum of the weighted edges, seen as vectors, cancel at every point. It is well known that this can be guaranteed for the interior vertices by applying a technique of Tutte. We show how to extend Tutte's ideas to construct a plane embedding where the weighted vector sums cancel also on the vertices of the boundary face. 
\end{abstract}

\section{Introduction}

\paragraph{Problem Setting.}
The graph of a polytope is an abstraction from its geometric realization.
 For a 3-polytope,
 the graph determines the complete combinatorial structure.
 The graphs of 3-polytopes are characterized by Steinitz' seminal
 theorem~\cite{s-emw-22}, which asserts that they 
 are exactly the planar 3-connected graphs.


A natural question is to ask for a geometric realization of a 3-polytope when its combinatorial structure is given. One might be interested in a realization that fulfills additionally certain optimality criteria. For example a good resolution is desirable to obtain aesthetic drawings~\cite{cgt-cdgtt-96,s-dpgvr-09}. 
We address a different problem and ask for an embedding whose vertices can be placed on a small integer grid. The vertex coordinates of such an embedding can be stored efficiently.

\paragraph{Related Work.}
Suppose we are given the combinatorial structure of a 3-polytope by
a graph $G$
with $n$ vertices.
The original proof of Steinitz' theorem transforms the
3-connected planar graph $G$ into the graph of the tetrahedron by a
sequence of elementary operations.
The transformation preserves the realizability as a 3-polytope. Since all operations can be carried out in the rationals, the proof gives a method to construct a realization of a 3-polytope with integer coordinates. However, it is not easy to keep track of the size and the denominators of the coordinates, which makes it difficult to apply this approach for our problem. An alternative proof of Steinitz' theorem goes back to the Koebe-Andreev-Thurston Circle Packing Theorem (see for example Schramm~\cite{s-eupsc-91}). This
approach relies on non-linear methods, which make the (grid) size of the embedding intractable. 
A third proof of Steinitz' theorem relies on the ``liftability''  of planar barycentric embeddings. Since this \emph{barycentric approach} is based on linear methods, its construction favors computational aspects of the embedding.  
This led to a series of embedding algorithms: Hopcroft and Kahn~\cite{hk-prga-92}, Onn and Sturmfels~\cite{os-qst-94}, Eades and Garvan~\cite{eg-dspgt-95}, Richter-Gebert~\cite{rg-rsp-96}, 
Chrobak, Goodrich, and Tamassia~\cite{cgt-cdgtt-96}.
Our work also follows this paradigm.

As a first quantitative analysis of Steinitz' theorem, Onn and Sturmfels~\cite{os-qst-94} showed that integer coordinates smaller than $n^{169n^3}$ suffice to realize a 3-polytope. 
Richter-Gebert improved this bound to $O(2^{18n^2})$. A more careful analysis of Richter-Gebert's approach shows that the size of the integer coordinates can be bounded by $2^{12n^2}$~\cite{r-rcpps-06}.
%

Integer realizations with at most exponentially large 
coordinates in terms of $n$ 
were previously known for polytopes whose graph contains a
triangle (Richter-Gebert  \cite{rg-rsp-96}). We describe this method in Section~\ref{sec:case1} (p.~\pageref{sec:case1}) as  Case~1
of our embedding algorithm.
In Richter-Gebert's approach (and already in
Onn and Sturmfels~\cite{os-qst-94}),
 graphs without triangles are embedded by first embedding the polar polytope, whose graph in this case has to contain a triangle. Based on the polar, an embedding of the original polytope is constructed. However, this
operation yields coordinates with a quadratic term in the exponent.

For \emph{triangulated} 
 3-polytopes,
Das and Goodrich~\cite{dg-copdc-97} showed that they can be embedded with coordinates of size $2^{O(n)}$, using an incremental method which can be carried out in $O(n)$ arithmetic operations.
Triangulated 3-polytopes are easier to realize on the grid than general polytopes, since each
vertex can be perturbed within some small neighborhood while
maintaining the combinatorial structure of the polytope. An explicit
bound on the coordinates has not been worked out by the authors.  For
stacked polytopes a better upper bound exists \cite{z-gcsg-07}, but it
is still exponential.

\paragraph{Lower Bounds.}
Little is known about the lower bound of a grid embedding of a 3-polytope. 
An integral convex embedding of an $n$-gon in the plane needs
an area of $\Omega(n^{3})$ \cite{az-mnecd-95,a-lbvsc-61,t-ep-91,vk-mnecd-82}.
Therefore, realizing a 3-polytope with an $(n-1)$-gonal face requires
at least one dimension of size $\Omega(n^{3/2})$.

\paragraph{Two Dimensions.}
In the plane, planar 3-connected graphs  can be embedded on a very small grid.
For a crossing-free straight-line embedding  an $O(n)\times O(n)$ grid
is sufficient \cite{fpp-hdpgg-90,s-epgg-90}.
This is also true if the embedding has to be convex
\cite{bfm-cdpg-06}. A
strictly convex drawing can be realized on an $O(n^2)\times
O(n^2)$ grid \cite{br-scdpg-06}.
\paragraph{Higher Dimensions.}
Already in dimension 4, there are polytopes that cannot be realized
with rational coordinates, and a 4-polytope that can be realized
on the grid might require coordinates that are doubly exponential in
the number of its vertices. Moreover, it is NP-hard to even decide
if a lattice is a face lattice of a
4-polytope~\cite{rg-rsp-96,gz-rspu-95}.

\paragraph{Results.}
In this article we develop an algorithm that realizes $G$ as a 3-polytope with
integer coordinates not greater than $O(187{.}13^n)=O(2^{7.55n})$. 
This implies that for any $3$-polytope a combinatorially  equivalent polytope can be stored
with $O(n)$ bits per vertex. 
For the case that $G$ contains a triangle we show that $G$ admits an integer realization with no coordinate larger than $O(28{.}\bar 4^n)=O(2^{4.82n})$, if $G$ contains a quadrilateral  face, the size of the coordinates can be bounded by $O(43.99^n)=O(2^{5.46n})$. The most difficult part of the algorithm is to locate the boundary face of the plane embedding such that a lifting into $\R^3$ exists. This problem can be reduced to a non-linear system which is most complex when $G$ contains neither a triangle nor a quadrilateral face.


Partial results containing
the essential ideas for graphs with quadrilateral faces (Case~2
of Section~\ref{sec:case2}) were presented by the second author
at the workshop \emph{The Future of Discrete Mathematics}
at {\v Sti\v r\'\i n} Castle, Czech Republic, in May 1997.
The results of this paper were presented in a different form
at the 23rd Annual Symposium on Computational Geometry
 in Gyeongju, Korea, in June 2007~\cite{rrs-epsg-07}. Since then, we were able to simplify the computation of the explicit bounds with help of Lemma~\ref{lem:zsize}. The simplification yields slightly different bounds. 
By improving the bound of Lemma~9 in~\cite{rrs-epsg-07} by a polynomial factor (now  Lemma~\ref{lem:foresttree}) 
we obtain better bounds in the end.
However, our analysis could be further improved with help of the more complicated construction of~\cite{rrs-epsg-07}. 
Since the improvement would only result in a constant factor we decided to present the simpler and more elegant analysis. 

 A follow-up work~\cite{s-dpgvr-09} extends the techniques of this article
and studies more general barycentric embeddings.
With help of these modifications, a grid embedding with $x$-coordinates
smaller than $2n$ can be constructed. The small $x$-coordinates are
realized at the expense of the size of the $y$ and $z$-coordinates,
which are bounded by $2^{O(n^2\log n)}$.

\paragraph{Remark:}
Most recently, Buchin and Schulz~\cite{BS10} improved the upper bound for the maximum number of spanning trees contained in a planar graph. 
This has a direct consequence for our results, since we obtain the bound for the necessary grid size in terms of this quantity. In particular, the new bounds of~\cite{BS10} yield that our algorithm requires a grid of size $O(147{.}71^n)=O(2^{7.21n})$ (general case), $O(39{.}87^n)=O(2^{5.32n})$ ($G$ contains a quadrilateral face), and $O(27{.}94^n)=O(2^{4.81n})$ ($G$ contains a triangular face).

%

\section{Lifting Planar Graphs}
\label{sec:lifting}

Let $G=(V,E)$ be a $3$-connected planar graph with vertex set $V=\{v_{1},\ldots,v_{n}\}$ embedded in the plane with straight edges and no crossings.
The coordinates of a vertex $v_{i}$ in the (plane) embedding are called $\p_{i}:=(x_{i},y_{i})^{T}$, the whole embedding is denoted as $G(\p)$. 
Let $h\colon V\to \R$ be a height assignment for the vertices in $G$. We write $z_{i}$ for $h(v_{i})$.
If the vertices $(x_{i},y_{i},z_{i})$ of every face of $G$ lie on a common plane, we call the height assignment $h$ a \textit{lifting} of $G(\p)$.

\begin{defn}[Equilibrium, Stress]
  An assignment $\omega \colon E\to \R$ of scalars \textup(denoted as
  $\omega(i,j)=\omega_{ij}=\omega_{ji}$\textup) to the edges of $G$ is
  called a \emph{stress}.
\begin{enumerate}
\item A vertex $v_i$ is in equilibrium in $G(\p)$, if
\begin{align}\label{equ:equilibrium}
\sum_{j:(i,j) \in E} \; \omega_{ij}(\p_i-\p_j) \;
= \; \mathbf{0}.
\end{align}
\item The embedding $G(\p)$ is in equilibrium if all vertices
  are in equilibrium.
\item If $G(\p)$ is in equilibrium for the stress $\omega$, then
  $\omega$ is called an equilibrium stress for $G(\p)$.
\end{enumerate}
\end{defn}
It is well known that equilibrium stresses and liftings are related.
 Maxwell observed
in the 19th century
 that there is a correspondence between
embeddings with equilibrium stress and projections of 3-dimensional
polytopes \cite{m-rfdf-64}. 
There are different versions of Maxwell's
theorem. For the scope of this article the following formulation is the most
suitable. 
\begin{thm}[Maxwell,
  Whiteley] Let $G$ be a planar 3-connected graph with embedding $G(\p)$ and 
designated face $f_1$. There exists a correspondence between
\begin{itemize}
\item[\rm A)] equilibrium stresses $\omega$ on $G(\p)$,
\item[\rm B)] liftings of $G(\p)$ in $\R^3$, where face $f_1$ lies in the
  $xy$-plane.
\end{itemize}
\end{thm}
The proof that A induces B (which is the important direction for our purpose) is  due to Walter Whiteley \cite{w-mspp-82}. The Maxwell-Cremona correspondence finds interesting applications in different areas (see for example
Hopcroft and  Kahn~\cite{hk-prga-92}, and Connelly, Demaine and Rote~\cite{cdr-spcpc-03}).

To describe a lifting, we have to specify for each face $f_{i}$ of the
graph 
the plane $H_{i}$ on which it lies.
 We define $H_i$ by the two
parameters $\mathbf{a}_i$ and $d_i$. The plane $H_i$ is characterized
by the function that assigns to every point $\p$ in the plane a third coordinate by
\begin{equation}
  \label{eq:plane}
H_i\colon \mathbf{p}\mapsto \langle \mathbf{p},\mathbf{a}_i\rangle +
d_i.  
\end{equation}
Here, $\langle \cdot ,\cdot \rangle$ denotes the dot product.
The correspondence between liftings and stresses comes from the
observation that the ``slope difference'' $\mathbf{a}_l-\mathbf{a}_r$
between two adjacent faces $f_l$ and $f_r$ is perpendicular to the
edge $\mathbf{p}_i-\mathbf{p}_j$ that separates them:
\begin{equation}
  \label{eq:perpendicular}
    \mathbf{a}_l-\mathbf{a}_r  =  \omega_{ij}
(\mathbf{p}_i-\mathbf{p}_j)^\perp,
\end{equation}
for some scalar $\omega_{ij}\in \mathbb{R}$.
Here, $\p^{\bot}:=\binom{-y}x$ denotes the vector $\p=\binom xy$
rotated by 90 degrees.
It is not hard to show that these numbers $\omega_{ij}$ form an
equilibrium stress.

The other direction, the computation of the lifting of $G(\p)$ induced by $\omega$ is straightforward,
see Crapo and Whiteley~\cite{cw-psspp-93}.
We follow the presentation of  Connelly, Demaine and Rote \cite{cdr-spcpc-03} for the computation of the lifting. 

The parameters $\mathbf{a}_i$ and $d_i$ can be computed by the following iterative method:
We pick $f_{1}$ as the face that lies in the $xy$-plane, and set $\mathbf{a}_1=\binom00$ and $d_0=0$. Then we lift the remaining faces one by one. This is achieved by selecting a face $f_{l}$ that is incident to an already lifted face $f_{r}$. 
Let $(i,j)$ be the common edge of $f_{l}$ and $f_{r}$. Assume that in $G(\p)$ the face  $f_{l}$ lies left of the  directed edge $ij$, and $f_{r}$ lies right of it. The parameters of $H_{l}$ can be computed by
 \begin{align}\label{equ:liftingit}
  \mathbf{a}_l & =  \omega_{ij}
(\mathbf{p}_i-\mathbf{p}_j)^\bot + \mathbf{a}_r, \\\label{equ:liftingit2}
d_l &=  \omega_{ij}\langle
\mathbf{p}_i,\mathbf{p}_j^\bot\rangle+d_r.
\end{align}
The formula~\eqref{equ:liftingit} comes directly from~\eqref{eq:perpendicular}, and
\eqref{equ:liftingit2} comes from the fact that the two planes must intersect above $\p_i$ and $\p_j$.

The sign of the stresses allows us to say something about the curvature of the lifted graph. 
According to~\eqref{eq:perpendicular} and \eqref{equ:liftingit}, the sign of $\omega_{ij}$ that separates $f_{l}$ and $f_{r}$ tells us if the lifted face $f_{l}$ lies below or above $H_{r}$.
As a consequence we obtain the following:
\begin{prop}\label{obs:signstress}
Let $G(\mathbf{p})$ be a straight-line embedding of a planar 3-connected graph
 $G$ with equilibrium stress. 
If the stresses on the boundary edges are
negative and all other stresses positive then the
lifting induced by
such equilibrium stress results in a convex 3-polytope. 
\end{prop}

\begin{lem}\label{lem:scaling1}
If $G(\p)$ has integer coordinates only and the equilibrium stress is integral on all interior edges, then the $z$-coordinates
of the lifted embedding are also integers.
\end{lem}
\begin{proof}
We select an interior face as face $f_{1}$.
The gradient $\mathbf{a}_1=(0,0)^T$ and
the scalar $d_1=0$ are clearly integral. 
For all other interior faces $f_{i}$ the parameters $\mathbf{a}_i$, $d_i$ 
of the planes $H_i$ can be computed with
help of equations~\eqref{equ:liftingit} and~\eqref{equ:liftingit2}. 
By an inductive argument 
these parameters are integral as well. 
Computing the $z$-coordinate of some point $\p_{i}$
by~\eqref{eq:plane}
boils down to the multiplication and addition of integers.
\end{proof}

\section{The Grid Embedding Algorithm}
\label{sec:embedding}
\subsection{The Plane Embedding}

The embedding of $G$ as a 3-polytope uses the following high level approach.
First we embed $G$ in the plane, such that it is liftable (see Section~\ref{sec:lifting}), then we lift the embedding to $\R^3$, finally we scale to obtain integer coordinates as described in Section~\ref{sec:scaling}. The analysis of the algorithm in Section~\ref{sec:analysis} gives the new upper bound. The most challenging part is to construct a liftable 2d-embedding. 



%


An embedding is called \emph{barycentric} if every vertex that is not
on the outer face is in the barycenter of its neighbors. Tutte showed
that for planar 3-connected graphs the barycentric embedding for a
fixed convex outer face is unique~\cite{t-crg-60,t-hdg-63}. Moreover,
if embedded with straight lines, no two edges cross, and all faces are
realized as convex polygons.  In the barycentric embedding all
vertices that are not on the outer face are in equilibrium according
to the stress $\omega\equiv 1$. Our embedding algorithm uses this
special stress only, although we state the lemmas as general as possible. Since our techniques might find applications in
other settings we develop our main tools for arbitrary
stresses. (Note that Tutte's approach works with
arbitrary stresses that are positive on interior
edges, see for example Gortler \etal~\cite{ggt-dofma-06}.)


We describe now how to compute the barycentric embedding of $G$. 
Let $f_0$ be a face of $G$ that we picked as the outer face, and let 
$k$ be the number of vertices in $f_0$. For simplicity we want $k$ as small as 
possible. 
Euler's formula implies that every planar and $3$-connected graph has a face $f_{0}$ with $k\leq 5$ edges. 
We assume that the vertices in $G$ are labeled such that the first $k$ vertices
belong to $f_0$ in cyclic order. 
Let $B:=\{1,\ldots,k\}$ be the index set of the boundary vertices and let $I:=\{k+1,\ldots,n\}$ denote the index set
of the interior  vertices.
The edges of $f_0$ are called \textit{boundary edges},
all other edges \textit{interior edges}. 
The stresses on the exterior edges 
 will be defined later, but since they don't matter for the barycentric embedding we set them to zero for now.

We denote with $L=(l_{ij})$ the \emph{Laplacian} matrix of $G$ (short \emph{Laplacian}), which is defined as follows
\[l_{ij}:=
\begin{cases} -\omega_{ij} & \mbox{if $(i,j)\in E$ and $i\not=j$},\\ 
\sum_{(i,j)\in E} \omega_{i,j} & \mbox{if $i=j$}, \\
 0 & \mbox{otherwise}.
\end{cases}
\]%
For the special ``weights'' $\omega\equiv 1$ the Laplacian equals the
negative adjacency matrix of $G$.
with vertex degrees on the diagonal.
We subdivide $L$ into block
matrices indexed by the  sets $I$ and $B$, and obtain $L_{IB}$, $L_{BI}$, $L_{BB}$, and $L_{II}$.
The matrix $L_{II}$ is called the \textit{reduced Laplacian matrix}  of $G$. For convenience we write $\bar L$ instead of $L_{II}$
\paragraph{Example.}
  \begin{figure}
    \centering
    \includegraphics{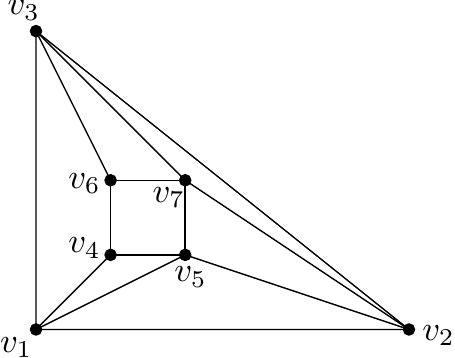}
    \caption{A small example graph.}
    \label{fig:smallexample}
  \end{figure}
Consider the graph of Figure~\ref{fig:smallexample},
with $B=\{1,2,3\}$ and $I=\{4,5,6,7\}$. We have
$$
L=
\left(\begin{BMAT}(@)[2.5pt]{rrr.rrrr}{bbb.bbbb}
  2&0&0&-1&-1&0&0\\
  0&2&0& 0&-1&0&-1\\
  0&0&2& 0& 0&-1&-1\\
 -1&0&0  & 3&-1&-1&0\\
 -1&-1&0 &-1& 4&0&-1\\
  0&0&-1 &-1& 0&3&-1\\
  0&-1&-1& 0&-1&-1&4\\
\end{BMAT}\right)
=
\begin{pmatrix}
  L_{BB} & L_{BI}\\
  L_{IB} & L_{II}\\
\end{pmatrix},
\text{ with }
\bar L=L_{II}=
\left(\begin{BMAT}(@)[2.5pt]{rrrr}{bbbb}
 3&-1&-1&0\\
-1& 4&0&-1\\
-1& 0&3&-1\\
 0&-1&-1&4\\
\end{BMAT}\right)
.$$
In the example the presence of the boundary edges $(1,2)$, $(1,3)$,
and $(2,3)$ is not reflected in the Laplacian, because the stress is set to zero on the boundary.

The location of the boundary vertices is given by the vectors 
$\mathbf{x}_B=(x_1,\ldots,x_k)^T$ and  $\mathbf{y}_B=(y_1,\ldots,y_k)^T$.
Since every vertex should lie at the (weighted) barycenter of its neighbors, the coordinates of the interior vertices   $\mathbf{x}_I=(x_{k+1},\ldots,x_n)^T$ and
  $\mathbf{y}_I=(y_{k+1},\ldots,y_n)^T$ have to satisfy the  equilibrium condition
\eqref{equ:equilibrium} for the stress $\omega$. In particular, 
the equations $\bar L \mathbf{x}_I + L_{IB}\mathbf{x}_B=\mathbf{0}$ and  $\bar L \mathbf{y}_I + L_{IB}\mathbf{y}_B=\mathbf{0}$ have to hold. 
Thus, we can express the interior coordinates as
\begin{equation}\label{equ:tuttex}
\begin{split}
  \mathbf{x}_I&=-\bar{L}^{-1}L_{IB}\mathbf{x}_B,\\
  \mathbf{y}_I&=-\bar{L}^{-1}L_{IB}\mathbf{y}_B.
\end{split}
\end{equation}
For non-zero weights $\omega$ the matrix $L_{II}$ is irreducible (that is, the underlying graph is connected, see Lemma~\ref{lem:technical}) and diagonally dominant. As a consequence $L_{II}$ is invertible and \eqref{equ:tuttex} has a unique solution \cite[page 363]{hj-ma-90}.

The barycentric embedding assures that the interior vertices are in equilibrium. However, to make $G(\p)$ liftable we have to guarantee the equilibrium also for the vertices on $f_0$. 
We define the vectors
$F:=\mathbf{F}_1,\ldots,\mathbf{F}_k$
as the non-resolving ``forces'', which arise at the boundary vertices and
cannot be canceled by the interior stresses:
\begin{align}\label{equ:nonresolve}
\forall i\in B \quad \sum_{(i,j) \in E } \; \omega_{ij}(\p_i-\p_j) \;
=: \mathbf{F}_i.
\end{align}
Our goal is to define the yet unassigned stresses on the boundary
edges such that they cancel the forces in $F$.  However, this is not
always possible, depending on the shape of the outer face.
In order to pick a good embedding of $f_0$, we have to know how
changing the coordinates of the outer face changes the forces in $F$.
The following lemma helps to express this dependence.

\index{substitution lemma} \index{stress!equilibrium}

\begin{lem}[Substitution Lemma] \label{lemma:1}
There are  weights $\ot_{ij} = \ot_{ji}$, for $i,j\in B$,
 independent of the location of the boundary vertices, 
 such that 
 \begin{align*}
\mathbf{F}_i= \sum_{j\in B:j\not=i}
 \ot_{ij}(\mathbf{p}_i-\mathbf{p}_j).
 \end{align*}
The weights $\ot_{ij}$ are the off-diagonal entries of  $L_{BI}\bar{L}^{-1}L_{IB}-L_{BB}$.
If $\omega$ is integral,
each $\ot$ is a  multiple of $1/ \det \bar{L}$.
\end{lem}
\begin{proof}
Let $\mathbf{F}_x$ denote the vector $(F_1^x,\ldots,F_k^x)^T$, 
where $F_i^x$ is the $x$-component of the vector $\mathbf{F}_i$.
We rephrase \eqref{equ:nonresolve} as   
$\mathbf{F}_x=L_{BB}\mathbf{x}_B+L_{BI}\mathbf{x}_I$.
With help of \eqref{equ:tuttex} we eliminate $\mathbf{x}_I$ and 
obtain
\[ \mathbf{F}_x=L_{BB}\mathbf{x}_B-L_{BI}\bar{L}^{-1}L_{IB}\mathbf{x}_B 
 =:\tilde{L}\mathbf{x}_B.
 \]
(The matrix $\tilde L=L_{BB}-L_{BI}\bar{L}^{-1}L_{IB}$ is the Schur complement of $\bar{L}$ in $L$.)
For the $y$-coordinates, we obtain a similar formula with the same matrix $\tilde{L}$. 
We define $\ot_{ij}$ as the off-diagonal entries $-\tilde{l}_{ij}$ of
$\tilde{L}$. Since $L_{BI}=(L_{IB})^T$, the matrix $\tilde{L}$ is
symmetric and therefore $\ot_{ij}=\ot_{ji}$ holds. 

To show that the expression
$\mathbf{F}_x=\tilde{L}\mathbf{x}_B$ has the form stated in the lemma we have
to check that all row sums in $\tilde{L}$ equal 0. Let $\mathbf{1}$
denote the vector  where all entries are $1$, equivalently $\mathbf{0}$ denotes the vector that contains only zeros as entries. Since each of the last $n-k$ rows of $L$ sums up to $0$ we have $\bar{L}\mathbf{1}+L_{IB}\mathbf{1}=\mathbf{0}$; and hence 
$-\bar{L}^{-1}L_{IB}\mathbf{1}=\mathbf{1}$. Plugging this
  expression into
  $\tilde{L}\mathbf{1}=L_{BB}\mathbf{1}-L_{BI}\bar L^{-1}L_{IB}\mathbf{1}$ gives us 
$\tilde{L}\mathbf{1}= L_{BB}\mathbf{1}+L_{BI}\mathbf{1}$, which equals
$\mathbf{0}$. 
The matrix $\tilde{L}$ can be written as a rational expression whose
denominator is the determinant of $\bar L$, and thus the
weights $\ot$ are multiples of $1/ \det \bar{L}$.
\end{proof}

In linear algebra terms, the lemma can be rephrased as saying that the
Schur complement of a submatrix of a weighted Laplacian, if it exists,
has again the form of a weighted Laplacian.

The proof assumes that $\bar L$ is invertible. This is the case
whenever the graph $G$ has no connected component that is a subset
of~$I$. If such components exist, they can simply be omitted, since
they are completely disconnected from $B$ and hence have no effect on
the forces in $F$. Hence, the lemma holds for arbitrary graphs, without any connectivity assumptions.

\paragraph{Example.}
For the example of Figure~\ref{fig:smallexample},
we obtain the following substitution stresses.
$$-\tilde{L}=
\begin{pmatrix}
  -\ot_{12}-\ot_{13}  &\ot_{12}&\ot_{13}\\
  \ot_{12}&-\ot_{12} -\ot_{23}&\ot_{23}\\
  \ot_{13}&\ot_{23} &-\ot_{13}-\ot_{23}
\end{pmatrix}
=
\begin{pmatrix}
  -96/95 & 3/5 & 39/95\\
  3/5 & -6/5 & 3/5\\
  39/95 & 3/5 & -96/95\\
\end{pmatrix}
$$

We emphasize that the $\ot$ values are independent
of the location of $\mathbf{x}_B$ and $\mathbf{y}_B$: they only depend on the
combinatorial structure of $G$. In other words, the stresses
$\ot_{ij}$ contain all the necessary information about the
combinatorial structure of $G$. Thus, we have a compact
description (of size $\binom k
2$) of the structure of $G$ that 
is responsible
for the forces in $F$. We call the stresses $\ot$ \textit{substitution
  stresses} \index{substitution stresses}
to emphasize that they are used
as a substitution for the combinatorial structure of~$G$.

For the later analysis of the grid size it is necessary to bound the
size of the substitution stresses. We first state a technical lemma.
\begin{lem}\label{lem:technical}
  Removing all vertices and edges of a face $f_{0}$ from a 3-connected
  planar graph $G$ leaves a connected graph.
\end{lem}
\begin{proof}
  After realizing $G$ as a polytope, the
  claim 
  becomes a special case of the well-known statement that a graph of a
  polytope in any dimension remains connected if the vertices of some
  face are removed. This statement can be proved by defining a linear
  objective function that realizes the minimum entirely on the removed
  vertices. Every remaining vertex is connected to the maximum vertex
  by a monotone increasing path.  The objective function
  can be perturbed such that there is a unique maximum.
\end{proof}
\begin{lem} \label{lem:2}
\begin{enumerate}
\item Let $\omega$ be a stress that is positive on every interior edge. Any induced substitution stress $\ot_{ij}$ is positive.
\item Let $\omega$ be the stress that is $1$ on every interior edge. Any induced substitution stress $\ot_{ij}$ is smaller than $n-k$.
\end{enumerate}
\end{lem}
\begin{proof}
The substitution stresses are independent of the location of $f_0$. 
Therefore, we can choose the positions for 
the boundary vertices freely.
We place vertex $v_i$ at position $(0,0)^T$ and all other boundary vertices at $(1,0)^T$.
All vertices lie on the segment between $(0,0)^T$ and $(1,0)^T$, which is the convex hull of the boundary vertices. 

We now show that all interior vertices lie in the interior of this
segment.
If an interior vertex lies at $(1,0)^T$ then all its neighbors have to lie at $(1,0)^T$ as well. Otherwise the vertex cannot be in equilibrium. But since due to Lemma~\ref{lem:technical} all interior vertices are connected by interior edges this would imply that all interior vertices must lie at $(1,0)^T$. This is a contradiction, since $v_{i}$ is also the neighbor of an interior vertex. By the same arguments one can show that no interior vertex can lie at $(0,0)^T$. 
Therefore, all interior vertices have a positive $x$-coordinate strictly smaller than $1$.

In our special embedding the force $\mathbf{F}_j$ ($j\not= i$) can be expressed as $\mathbf{F}_j=(\ot_{ij},0)^T$.
By \eqref{equ:nonresolve} we 
have $\ot_{ij}=\sum_{k \in I}\omega_{jk}(x_j-x_k)$. 
Due to the results of the previous paragraph, this sum consists of at most $|I|$  summands, which are positive, and in the case $\omega\equiv 1$ smaller than 1. Both statements of the lemma follow.
\end{proof}

We are now ready to introduce the embedding algorithm. 
As a first step we construct a 2d
embedding in equilibrium with respect to a stress $\omega$ with $\omega\equiv 1$ on the interior edges.
In order to get equilibrium on the boundary vertices as well,
we have to choose their locations and the stresses on the boundary edges appropriately.
This leads to a
non-linear system in the $2k$ unknowns $\mathbf{x}_{B}$, and $\mathbf{y}_{B}$ and the $k$
unknown boundary stresses $\omega_{12}, \omega_{23},\dotsc , \omega_{k1}$.
Let $L_0$ be the Laplacian of the graph that consists of the outer face $f_0$ only, with unknown stresses
 $\omega_{12},\omega_{23},\ldots,\omega_{k1}$
for the boundary edges. The $2k$ equations of the system are given by 
\begin{align}\label{equ:system}
L_0\mathbf{x}_B + \tilde L \mathbf{x}_B = \mathbf{0}, \quad  L_0\mathbf{y}_B + \tilde L \mathbf{y}_B = \mathbf{0}.
\end{align}
Since these equations are dependent, the system is
under-constrained. To solve it, we fix as many boundary coordinates as
necessary to obtain a unique solution.  We also have to ensure that the
solution defines a \emph{convex} face. We continue with a case
distinction on $k$.

\subsubsection*{Case 1: $G$ contains a triangular face}\label{sec:case1}
The triangular case is easy: we can position the boundary  vertices at any
convenient position (see for example~\cite{hk-prga-92}). We choose:
\begin{align}\label{triangle}
 \mathbf{p}_1= \begin{pmatrix}0 \\  0\end{pmatrix},
 \mathbf{p}_2= \begin{pmatrix}1 \\  0\end{pmatrix},
 \mathbf{p}_3= \begin{pmatrix}0 \\ 1 \end{pmatrix}.
\end{align}
\begin{lem}
If $G$ contains a triangle and we place the boundary
vertices as
stated in \eqref{triangle} then the boundary forces can be resolved.
\end{lem}
\begin{proof}
We embed $G$ as barycentric embedding and calculate the substitution
stresses. 
After setting $\omega_{12}=-\ot_{12},\omega_{23}=-\ot_{23},\omega_{13}=-\ot_{13}$
all points are in equilibrium. 
\end{proof}

\subsubsection*{Case 2: $G$ contains a quadrilateral but no triangular face}\label{sec:case2}
If $f_0$ is a quadrilateral we have to fix some 
coordinates of the boundary vertices such that
it is possible to cancel the forces in $F$. 
We used computer algebra software to experiment with various
possibilities to constrain the coordinates and solve the non-linear
system~\eqref{equ:system}.
A unique solution can be obtained by setting 
\begin{align}\label{equ:4goncoordinates}
 \mathbf{p}_1= \begin{pmatrix}0 \\  0\end{pmatrix},
 \mathbf{p}_2= \begin{pmatrix}1 \\  0\end{pmatrix},
 \mathbf{p}_3= \begin{pmatrix}2 \\ y_3 \end{pmatrix},
 \mathbf{p}_4= \begin{pmatrix}0 \\ 1 \end{pmatrix},
\end{align}
with
\begin{align}
y_3 = \frac{\ot_{24}}{2\ot_{13}-\ot_{24}}.
\label{equ:sol4}
\end{align}
The solution of the equation system~\eqref{equ:system} provides also the stresses on the boundary edges. These stresses are not necessary for our further computations; we mention them here for completeness only.
\begin{equation}
\begin{split}
\omega_{12}& =  -2\ot_{13}-\ot_{12}, \\
\omega_{23}& =  \ot_{24}-2\ot_{13}-\ot_{23}, \\
\omega_{34}& =  -\frac{\ot_{24}}{2}-\ot_{34}, \\
\omega_{14}& =  \frac{\ot_{24}\ot_{13}}{\ot_{24}-2\ot_{13}}-\ot_{14}.
\end{split}
\label{equ:sol4-2}
\end{equation}
We assume that $\ot_{13}\geq\ot_{24}$.
(Otherwise we cyclically relabel the vertices on $f_0$.) 
Since $\omega\equiv 1$ on the interior edges the substitution stresses are positive by Lemma~\ref{lem:2}. Under this assumption we can deduce that $0<y_3\leq 1$. Hence,  $f_0$ forms a \emph{convex} face.

Note that the substitution stresses $\tilde \omega_{ij}$ between adjacent
vertices (on the boundary) are irrelevant. The 
forces resulting by the boundary stresses $\ot_{ij}$ 
can be directly canceled by the
corresponding stresses $ \omega_{ij}$. 
This can also be observed by looking at the solution  
of the corresponding equation system:
Boundary stresses do not appear in the solution for $y_3$. (Furthermore 
the sum $\ot_{ij}+\omega_{ij}$ for boundary edges $(i,j)$
does not depend on any other boundary stress either.)

\begin{lem}
If $G$ contains a quadrilateral and we place the boundary vertices as
stated in \eqref{equ:4goncoordinates} and \eqref{equ:sol4}, 
then $f_0$ forms a convex quadrilateral 
and the boundary stresses
 \eqref{equ:sol4-2}
 cancel the
forces~$F$.
\qed
\end{lem}

\noindent
\subsubsection*{Case 3: $G$ contains no triangular and no quadrilateral
  face}
The case if the smallest face of $G$ is a  pentagon is more complicated. 
We have $\binom52=10$ substitution stresses $\tilde \omega_{ij}$, but the adjacent ones
do not count (by the same reasons given in the previous case).
So we are left with five ``diagonal'' substitution stresses
$\ot_{13},\ot_{14},\ot_{24},\ot_{25}$, and $\ot_{35}$. 

Like in the previous cases we determine a unique solution of the equation system
by fixing some of the coordinates of the outer face. 
However, we have to make more effort to guarantee the
convexity of $f_0$. We first observe:
\begin{lem}\label{lem:lemma5gon}
We can relabel the boundary points for any stress $(\ot_{ij})_{1\leq i,j
  \leq 5}$ such that
\[ \ot_{35}  \geq  \ot_{24}  \quad\mbox{and}\quad  \ot_{25}  \geq  \ot_{13}.\]
\end{lem}
\begin{proof}
Without loss of generality we assume that the largest stress on an
interior edge is $\ot_{35}$. If
$\ot_{25}\geq\ot_{13}$ we are done. Otherwise we relabel the vertices by
exchanging $\mathbf{p}_3 \leftrightarrow \mathbf{p}_5$ and
$\mathbf{p}_1 \leftrightarrow \mathbf{p}_2$.
\end{proof}
For the rest of this section we label the vertices such that Lemma~\ref{lem:lemma5gon} holds.
The way we embed  $f_0$ depends on the substitution stresses $\ot_{ij}$.\\[0.5ex]
\textbf{Case 3A:}\\
We assume that
\begin{align}\label{equ:case1}
\ot_{35}\ot_{14}+\ot_{14}\ot_{25}+\ot_{25}\ot_{24}+\ot_{13}\ot_{35}>\ot_{35}\ot_{25}.
\end{align}
In this case we assign
\[
 \mathbf{p}_1= \begin{pmatrix}0 \\  0\end{pmatrix},
 \mathbf{p}_2= \begin{pmatrix}1 \\  0\end{pmatrix},
 \mathbf{p}_3= \begin{pmatrix}1 \\ 1 \end{pmatrix},
 \mathbf{p}_4= \begin{pmatrix}0 \\ 1 \end{pmatrix},
 \mathbf{p}_5= \begin{pmatrix}x_5\\y_5 \end{pmatrix}.
\]
\begin{figure}[ht]
 \center 
 \begin{tabular}{cp{1cm}c}
   \includegraphics[scale=.7]{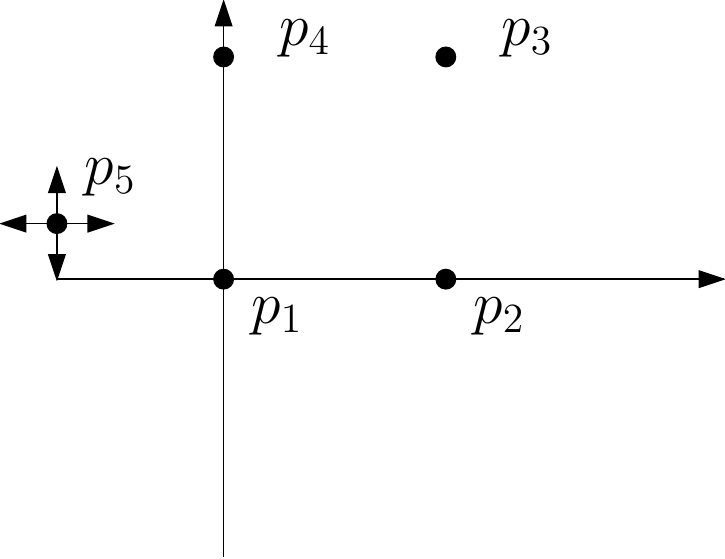} & &
   \includegraphics[scale=.7]{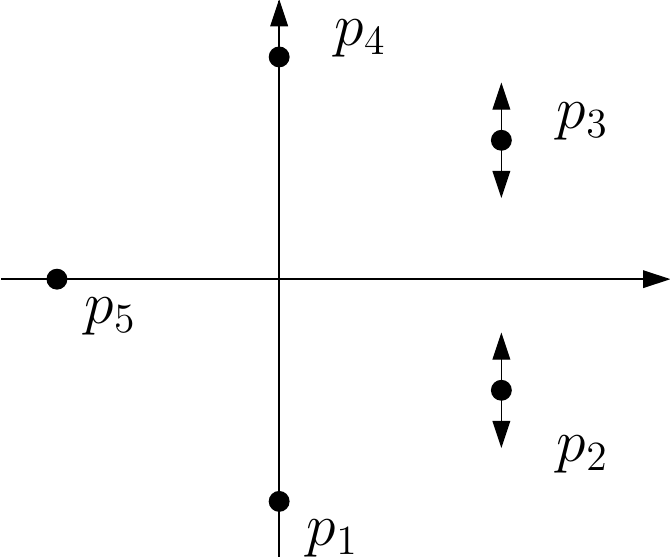} \\
  (a) Case 3A &&  (b) Case 3B 
 \end{tabular}
    \caption{Placement of the boundary vertices.}
    \label{fig:placement}
\end{figure}
Figure~\ref{fig:placement}(a) illustrates the location of the points.
Together with the equations of 
\eqref{equ:system} we obtain as solution for $\mathbf{p}_5$:
\[x_5=
 \frac{(\ot_{13}-\ot_{25}-\ot_{24})(\ot_{35}+\ot_{13}-\ot_{24})}       {\ot_{35}\ot_{14}+\ot_{14}\ot_{25}+\ot_{25}\ot_{24}+\ot_{13}\ot_{35}-\ot_{35}\ot_{25}},\]
 \[y_5=\frac{\ot_{35}+\ot_{13}-\ot_{24}}{\ot_{35}+\ot_{25}} .\]
The boundary stresses $\omega_{ij}$ are complicated expressions
that are not necessary for further computations.
We list them here for completeness only.
\begin{align*}
\omega_{12}  &= 
 \frac{\ot_{13}(\ot_{25}^2+\ot_{24}\ot_{35}+2\ot_{24}\ot_{25}-\ot_{13}\ot_{25})
       \!+\!\ot_{14}(\ot_{25}^2+\ot_{25}\ot_{35}+\ot_{24}\ot_{25}+\ot_{35}\ot_{24})}
       {\ot_{35}\ot_{25}-\ot_{14}\ot_{25}-\ot_{25}\ot_{24}-\ot_{13}\ot_{35}-\ot_{35}\ot_{14}}- \ot_{12},\\         
\omega_{34} &=  
 \frac{ \ot_{14}(\ot_{35}^2+\ot_{35}\ot_{13}+\ot_{25}\ot_{35}+\ot_{13}\ot_{25})
       \!+\!\ot_{24}(\ot_{35}^2+\ot_{13}\ot_{25}+2\ot_{13}\ot_{35}-\ot_{35}\ot_{24})}
       {\ot_{35}\ot_{25}-\ot_{14}\ot_{25}-\ot_{25}\ot_{24}-\ot_{13}\ot_{35}-\ot_{35}\ot_{14}} - \ot_{34},  \\
\omega_{23} & = 
\frac{\ot_{13}\ot_{25}+\ot_{25}\ot_{35}+\ot_{24}\ot_{25}}
         {-\ot_{25}-\ot_{35}}     -\ot_{23}, \\
\omega_{45}  & = 
\frac{\ot_{24}\ot_{25}+\ot_{25}\ot_{14}+\ot_{14}\ot_{35}+\ot_{24}\ot_{35}}
         {\ot_{13}-\ot_{24}-\ot_{25}}     -\ot_{45}, 
\\[1.1ex]
\omega_{15} & = 
\frac{\ot_{13}\ot_{25}+\ot_{35}\ot_{13}+\ot_{14}\ot_{25}+\ot_{14}\ot_{35}}
         {\ot_{24}-\ot_{35}-\ot_{13}}     -\ot_{15}. 
\end{align*}

We have to check that $f_0$ forms a convex polygon. 
Clearly, $y_5>0$, since
the $\ot_{ij}$'s  are greater than zero and $\ot_{35}\geq\ot_{24}$. Moreover 
$y_5<1$, because $\ot_{25}\geq\ot_{13}$. The numerator of
$x_5$ is negative and due to \eqref{equ:case1} the denominator
of $x_5$ is positive. Therefore, $x_5<0$ and $f_0$ forms
a convex polygon.\\[0.5ex]
\textbf{Case 3B:}\\
We assume the opposite of \eqref{equ:case1}, namely 
\begin{align}\label{equ:case2}
 \ot_{35}\ot_{14}+\ot_{14}\ot_{25}+\ot_{25}\ot_{24}+\ot_{13}\ot_{35}\leq\ot_{35}\ot_{25}.
\end{align}
The coordinates for the boundary vertices are chosen as
\[\mathbf{p}_1\!=\!
\begin{pmatrix} 0 \\ -1 \end{pmatrix},\mathbf{p}_2\!=\!\begin{pmatrix} 1 \\ y_2\end{pmatrix},
\mathbf{p}_3\!=\!\begin{pmatrix} 1 \\ y_3 \end{pmatrix} ,\mathbf{p}_4\!=\!\begin{pmatrix} 0 \\
  1\end{pmatrix},\mathbf{p}_5\!=\!\begin{pmatrix} -1 \\ 0\end{pmatrix}.\] 
See Figure~\ref{fig:placement}(b) for an illustration.
This leads to the solution 
\begin{align*}
 y_2& =  -2\cdot \frac{\ot_{24}\ot_{13}+\ot_{24}\ot_{35}+
 \ot_{25}\ot_{13}+2\ot_{25}\ot_{35}-\ot_{13}^2-2\ot_{13}\ot_{35}
 -\ot_{35}\ot_{14}}
 {\ot_{24}\ot_{35}+\ot_{25}\ot_{13}+2\ot_{25}\ot_{35}}, 
 \\
 y_3&=\phantom{-} 2\cdot\frac{\ot_{24}\ot_{13}+\ot_{24}\ot_{35}+
 \ot_{25}\ot_{13}+2\ot_{25}\ot_{35}-\ot_{24}^2-2\ot_{24}\ot_{25}
 -\ot_{14}\ot_{25}}
 {\ot_{24}\ot_{35}+\ot_{25}\ot_{13}+2\ot_{25}\ot_{35}}.
\end{align*}
The boundary stresses   $\omega_{ij}$ are once more  not necessary for further computations and listed for completeness only.
\begin{align*}
\omega_{12} & =  -\ot_{24}-2\ot_{25}-\ot_{12}, \\[1.1ex]
\omega_{23}& =   \frac{-\ot_{25}(\ot_{13}^2+2 \ot_{13}\ot_{35}+2\ot_{24}\ot_{35})
             -\ot_{35}(\ot_{24}^2+\ot_{25}\ot_{14})}
             {2\ot_{35}(\ot_{24}\!+\!\ot_{25}\!-\!\ot_{13}\!-\!\frac{1}{2}\ot_{14})
             \!+\!2\ot_{25}(\ot_{13}\!+\!\ot_{35}\!-\!\ot_{24}\!-\!\frac{1}{2}\ot_{14})
             \!-\!(\ot_{13}\!-\!\ot_{24})^2}-\ot_{23},\\[1.1ex] 
\omega_{34} & =
  -\ot_{14}-2\ot_{15}-\ot_{34},  \\[1.1ex]
\omega_{45} & =
  \ot_{24}-2\ot_{35}-\ot_{13}-\ot_{45}, \\[1.1ex]
\omega_{15} &=
  \ot_{13}-2\ot_{25}-\ot_{24}-\ot_{15}. 
\end{align*}

The outer face is convex if $-2<y_2<y_3<2$. 
The inequalities $-2<y_2$ and $y_3<2$ are equivalent to
\begin{align*}
-\ot_{13}^2-\ot_{35}\ot_{14}+\ot_{13}(\ot_{24}-2\ot_{35})&<0 
 \text{ and} \\
-\ot_{24}^2-\ot_{14}\ot_{25}+\ot_{24}(\ot_{13}-2\ot_{25})&<0.
\end{align*}
Both inequalities hold, because we add only negative summands on the left side.
It remains to check if $y_2-y_3<0$. First we get rid of the
denominator and bring all negative summands on the right side. This
leads to the equivalent inequality
\begin{multline}\label{equ:case2convex}
\ot_{13}^2+\ot_{24}^2+2\ot_{13}\ot_{35}+2\ot_{24}\ot_{25}+ \ot_{25}\ot_{14}+\ot_{35}\ot_{14}<\\
  2\ot_{24}\ot_{35}+2\ot_{25}\ot_{13}+4\ot_{25}\ot_{35}+2\ot_{24}\ot_{13}. 
\end{multline}
We observe that $\ot_{13}^2\leq\ot_{25}\ot_{13}$ and
$\ot_{24}^2\leq\ot_{24}\ot_{35}$. Because of  the assumption~\eqref{equ:case2} we have 
$4\ot_{35}\ot_{25}>2\ot_{13}\ot_{35}+2 \ot_{24}\ot_{25}+\ot_{25}\ot_{14}+\ot_{35}\ot_{14}$.
Therefore, the right side of \eqref{equ:case2convex} is greater than its left
side, which shows that $y_2<y_3$ and $f_0$ forms a convex pentagon.
This completes the case distinction and we conclude with the following lemma.
\begin{lem}
If we place the boundary vertices as
discussed above, then the outer face will be embedded as a convex pentagon and 
the computed boundary stresses cancel the forces in $F$.
\qed
\end{lem}

We have defined four different ways to embed $G$. The selected embedding
depends on the combinatorial structure $G$. If $G$ contains a
triangular face
we say it is of \textit{type 3}.\index{graph!type} If it contains a quadrilateral but
no triangular face $G$ is of \textit{type 4}. Otherwise the embedding
depends on the substitution stresses induced by the combinatorial structure of $G$. If
\eqref{equ:case1} holds (Case 3A) $G$ is
of \textit{type 5A}, otherwise (Case 3B) we say $G$ is of \textit{type 5B}. 

\subsection{Lifting the Plane Embedding and Scaling to Integrality}\label{sec:scaling}

We continue with lifting the plane embedding of
$G$ to $\mathbb{R}^3$. With help of the observations made in Section~\ref{sec:lifting} the incremental computation of the 3d embedding is straightforward. It suffices to compute for every face $f_{i}$ the corresponding plane $H_{i}$.

Since the embedding of $f_0$ is convex, 
the boundary stresses must necessarily be negative,
since otherwise the boundary vertices could not be in equilibrium with all
interior stresses being positive.
Thus we do not need to explicitly check the sign of the boundary stresses. The sign pattern of 
the stress implies that the lifting of the plane embedding gives a convex
polytope (see Proposition~\ref{obs:signstress}).

As described in Section~\ref{sec:lifting},
we begin the lifting by fixing the plane $H_1$ for some interior face $f_1$ as the $x$-$y$-plane. We set
$\mathbf{a}_1=(0,0)^T$, $d_1=0$,
and compute the remaining planes face by face using equations~\eqref{equ:liftingit}
 and~\eqref{equ:liftingit2}.
 It is not necessary to compute the parameters
of $H_0$ since we can determine the heights of $\p_1,\ldots,\p_k$ by some plane $H_i$ of
an interior face. Hence, the lifting can be computed using only stresses on interior edges.
This simplifies the later analysis because all interior 
stresses are $1$, whereas the boundary stresses are complicated expressions.




It can be observed that the computed lifting has rational coordinates. This is true because the barycentric embedding gives rational coordinates and the lifting process is based on multiplication and addition of the 2d coordinates. Hence, the $z$-coordinates are also rational.
We analyze the common denominator of the coordinates to obtain scaling factors for the integral embedding. We use different scaling factors $S_x$ for 
the $x$-coordinates and $S_y$ for the $y$-coordinates.

As a consequence of Lemma~\ref{lem:scaling1} it is sufficient to 
scale to integer $x$ and $y$-coordinates.
Furthermore we observe:\index{reduced Laplacian matrix}
\begin{lem}\label{lem:cramer}
If the boundary points are integral, the barycentric embedding yields
coordinates that are multiples of $1/ \det \bar L$.
\end{lem}
\begin{proof}
The interior plane coordinates are a result of equation~\eqref{equ:tuttex}.
By Cramer's rule every coordinate can be expressed as 
\[x_i=\det{\bar{L}^{(i)}} / \det{\bar{L}},\]
where $\det{\bar{L}^{(i)}}$ is obtained from $\bar{L}$ by replacing the
$i$-th column of $\bar{L}$ by $L_{IB}\mathbf{x}_B$. 
Since $\det{\bar{L}^{(i)}}$ is integral,  $\det{\bar{L}}$ is the denominator of $x_i$.
The same holds for~$y_i$.
\end{proof}

Our first goal is to scale the plane embedding such that the boundary vertices 
get integer coordinates.  
Let $S^B_x$ be the integral scaling factor that gives integer \textit{boundary} $x$-coordinates
and  $S^B_y$ be the integral scaling factor that gives integer \textit{boundary} $y$-coordinates.
Due to Lemma~\ref{lem:cramer} the scaling factors
$S_x:=S_x^B\det{\bar{L}}$ and $S_y:=S_y^B\det{\bar{L}}$ make all vertices integral.
Since we choose integral scaling factors $S_x^B$ and $S_y^B$ no integer 
coordinate is scaled to a non-integer.

Let us now compute the factors that are necessary to scale to integer 
boundary coordinates.
Clearly the scaling factors depend on the type of $G$. 
If $G$ is of type $3$ then we need not 
scale, since all boundary coordinates are either $0$ or $1$. 
If $G$ is of type $4$ we have to scale the
$y$-coordinates only (see~\ref{equ:sol4}). 
We multiply $y_{3}$ with 
$S_y^B:=(2\ot_{13}-\ot_{24})\det{\bar{L}}$, which gives 
$S_y^B y_{3}= \ot_{24}\det{\bar{L}}$,  which due to the Substitution Lemma is an integer.

If $G$ is of type $5A$ we have to scale such that $x_5$ and $y_5$
 become integral. We pick
\begin{align*}
 S_x^B & = 
 {(\ot_{35}\ot_{14}+\ot_{14}\ot_{25}+\ot_{25}\ot_{24}+\ot_{13}\ot_{35}-\ot_{35}\ot_{25}})
 (\det{\bar{L}})^2,
 \\
 S_y^B & =  (\ot_{35}+\ot_{25})  \det{\bar{L}}. 
 \end{align*}
It can be easily checked that these 
factors as well as  $S^B_x x_5$ and $S^B_y y_5$ are integral.

When $G$ is of type $5B$, 
the only non-integer boundary coordinates are $y_2$ and $y_3$, we need
to scale in $y$-direction only. We choose
\[S^B_y=(\ot_{24}\ot_{35}+\ot_{25}\ot_{13}+2\ot_{25}\ot_{35}) (\det{\bar{L}})^2.\]
Again, due to the Substitution Lemma, $S^B_y$, $S^B_y y_2$, and $S^B_y y_3$ are all integral.

For every type of $G$ there is a pair of scaling factors $S_{x},S_{y}$, such that the
scaled boundary points are integral. 
Table~\ref{tab:scaling} summarizes the discussion and lists the final
scaling factors depending on the type of $G$.
\renewcommand{\arraystretch}{1.2}
\begin{table}[htb]
   \centering
   \begin{tabular}{cl}
\hline \hline
       type of $G$ & scaling factors \\
     \hline\hline
3    & 
$S_x=
S_y=\det\bar{L}$\\\hline
      &$S_x=\det{\bar{L}}$  \\ 
     \raisebox{1.5ex}[1.5ex]{$4$}     &       $S_y=(2\ot_{13}-\ot_{24})(\det{\bar{L}})^2$\\\hline
      & $S_x={(\ot_{35}\ot_{14}+\ot_{14}\ot_{25}+\ot_{25}\ot_{24}+\ot_{13}\ot_{35}-\ot_{35}\ot_{25}})
 (\det{\bar{L}})^3$ \\ 
     \raisebox{1.5ex}[1.5ex]{$5A$}     &
     $S_y=(\ot_{35}+\ot_{25})  (\det{\bar{L}})^2$ \\ \hline
      & $S_x= \det{\bar{L}}$\\ 
     \raisebox{1.5ex}[1.5ex]{$5B$}     & $S_y=(\ot_{24}\ot_{35}+\ot_{25}\ot_{13}+2\ot_{25}\ot_{35}) 
(\det{\bar{L}})^3$\\ \hline
 \end{tabular}
   \caption{The scaling factors $S_x$ and $S_y$ for the
     different types of $G$.}
 \label{tab:scaling}
 \end{table}

\subsection{Analysis of the Grid Size}
\label{sec:analysis}
To bound the size of the coordinates of the integer embedding it is crucial
to obtain a good bound for $\det \bar L$. 
Recall that we assume unit stresses $\omega \equiv 1$
on the interior edges, throughout.
There exists a connection between the
number of spanning trees in $G$ and $\det\bar L$. Let us first define:
\begin{defn}
 Let $\mathcal{B}$ be a subset of vertices of $G$. A subgraph of $G$ is \index{B-forest@$\mathcal{B}$-forest}
 called \emph{spanning $\mathcal{B}$-forest} if 
 \begin{itemize}
 \item it consists of $|\mathcal{B}|$ vertex disjoint trees covering
   all vertices of $G$,
 \item each tree contains a unique vertex from $\mathcal{B}$.
 \end{itemize}
\end{defn}
In the following we use the set of boundary vertices  for  $\mathcal{B}$.
Let $\mathcal{F}_B(G)$ denote the number of spanning $B$-forests of $G$ and 
$\mathcal{T}(G)$ the number of spanning trees of $G$.
A generalization of the Matrix-Tree Theorem\index{Matrix-Tree-Theorem}
\cite{l-gmtt-82} (see also
\cite{r-rcpps-06}) states that the number of spanning $B$-forests of $G$
is $\det\bar L$.
In our case, we can directly bound $\mathcal{F}_B(G)$ by $\mathcal{T}(G)$. 
\begin{lem}\label{lem:foresttree}
Let $G$ be a planar graph with a distinguished face and let $B$ be the
set of vertices of this face. 
The number of spanning $B$-forests of $G$ is bounded
from above by
\[  \mathcal{F}_B(G) <  \mathcal{T}(G).\]
\end{lem}
\begin{proof}
Every spanning $B$-forest can be turned into a spanning tree by adding all
boundary edges except $(1,2)$. No two distinct spanning $B$-forests
are associated with the same spanning tree. Therefore, the number
of spanning trees exceeds the number of spanning $B$-forests.
Since there is a spanning tree that contains the edge $(1,2)$ the inequality is strict.
\end{proof}

It is easy to give an exponential upper bound for $\mathcal{T}(G)$:
\begin{prop}[Rib\'o~Mor \cite{r-rcpps-06}]
\index{spanning trees!in a planar graph}
\begin{enumerate}
\item The number of spanning trees in a graph is bounded by the product of all vertex degrees:
\[\mathcal{T}(G) <  \prod_i \mathrm{deg}(v_i).\]
\item For a planar graph, we have $\mathcal{T}(G) < \prod_i \mathrm{deg}(v_i)< 6^n$.
\end{enumerate}
\end{prop}
\begin{proof}
1. 
Consider all
 directed graphs that are obtained by choosing
an outgoing edge in $G$ out of every vertex except $v_n$. The number
of these directed graphs is given by $\prod_{i=1}^{n-1} \mathrm{deg}(v_i)$.
By ignoring the edge orientations, one obtains all spanning trees
(and many graphs that are not spanning trees). 
Alternatively, the bound can be proved by applying a variant of Hadamard's
inequality for positive semidefinite matrices 
\cite[page 176]{z-mt-99} to the (positive semidefinite) matrix $L'$
that is obtained by removing from the Laplacian $L$ the row and
column corresponding to the 
vertex $v_n$:
\begin{align*}
 \mathcal{T}(G)&
=\det\bar L
\le 
\prod_{i=1}^{n-1} l_{ii}
=\prod_{i=1}^{n-1} \deg(v_i)
\end{align*}

2. This follows from the arithmetic-geometric-mean inequality and the
fact that 
$\sum_i \mathrm{deg}(v_i)< 6n$, which is a consequence of
Euler's formula.
\end{proof}
Sharper bounds for $\mathcal{T}(G)$ have been given by
{Rib\'o~Mor} et al.~\cite{rry-abnst-09}, see also
\cite{r-rcpps-06,r-nstpg-05}.
These bounds take into account whether $G$ contains triangular or
quadrilateral faces:
\begin{align*}
\mbox{if $G$ is of type 3:} \quad  \mathcal{F}_B(G) <  \mathcal{T}(G) & \leq   5.\bar{3}^n, \\
\mbox{if $G$ is of type 4:} \quad  \mathcal{F}_B(G) <  \mathcal{T}(G) & \leq  3.529988^n,   \\
\mbox{if $G$ is of type 5A/5B:} \quad  \mathcal{F}_B(G) <  \mathcal{T}(G) & \leq  2.847263^n.
\end{align*}

Since we know upper bounds for the $\ot$ values (by Lemma~\ref{lem:2}) and $\det \bar L$ (by 
the previous discussion) we can bound the size of the integer coordinates of the embedding
of $G$. We start with bounding the $x$ and $y$-coordinates. Let $\Delta x$ 
denote an upper bound for 
the difference between the largest and the smallest $x$-coordinate. $\Delta y$ is defined
in the same way for the $y$-coordinates.

Again we have to discuss the $4$ cases separately. If $G$ is of type $3$ then clearly
$\Delta x=\Delta y = \det \bar L$. If $G$ is of type $4$ the largest $x$-coordinate is $2 S_x$ 
and the smallest zero. Thus we have $\Delta x=2 \det \bar L$. The largest $y$-coordinate is
obtained at $y_4=1$ (remember $y_{3}\leq 1$), 
therefore $\Delta y=S_y=(2\ot_{13}-\ot_{24})(\det\bar L)^2$. Let us now
assume $G$ is of type $5A$. 
The value of $\Delta x$ is given by $x_2-x_5$. Evaluating this expression leads to
\[\Delta x = (\ot_{25}(\ot_{13}+\ot_{14})+\ot_{35}(\ot_{14}+\ot_{25})-(\ot_{13}-\ot_{24})^2)(\det\bar L)^3.\]
Since the smallest $y$-coordinate is zero we have $\Delta y=y_3$, which equals 
$(\ot_{35}+\ot_{25})(\det\bar L)^2$. It remains to discuss the case when $G$ is of type $5B$. 
\textit{Before} the scaling the coordinates fulfill $-1\leq x \leq 1$ and $-2<y<2$.
Combining these inequalities with the scaling factors yields $\Delta x =2\det \bar L$ and 
$\Delta y= 4(\ot_{24}\ot_{35}+\ot_{25}\ot_{13}+2\ot_{25}\ot_{35}) 
(\det{\bar{L}})^3$. We sum up the results for $\Delta x$ and $\Delta y$ in Table~\ref{tab:dx} and 
Table~\ref{tab:dy}. With help of Lemma~\ref{lem:2} we can eliminate the $\ot$ values that appear in the bounds of $\Delta x$ and $\Delta y$. The resulting upper bounds, which we use in the further analysis, 
are listed in Table~\ref{tab:upperbounds}. 
\renewcommand{\arraystretch}{1.2}
\begin{table}
   \centering
   \begin{tabular}{cl}
       \hline \hline
       type of $G$ & $\Delta x$  \\
     \hline\hline
     $3$ & $\det \bar L$  \\
     $4$ & $2\det \bar L$  \\ 
     $5A$ & $(\ot_{25}(\ot_{13}+\ot_{14})+\ot_{35}(\ot_{14}+\ot_{25})-(\ot_{13}-\ot_{24})^2)(\det\bar L)^3$\\
    $5B$  & $2\det{\bar{L}}$ \\
 \hline
   \end{tabular}
   \caption{The values $\Delta x$ depending on the type of $G$.}
 \label{tab:dx}
 \end{table}
 
\renewcommand{\arraystretch}{1.2}
\begin{table}
   \centering
   \begin{tabular}{cl}
       \hline \hline
       type of $G$ & $\Delta y$  \\
     \hline\hline
     $3$ &  $\det \bar L$ \\
     $4$ &  $(2\ot_{13}-\ot_{24})(\det\bar L)^2$ \\ 
    $5A$ & $(\ot_{35}+\ot_{25})(\det\bar L)^2$\\
    $5B$ & $4(\ot_{24}\ot_{35}+\ot_{25}\ot_{13}+2\ot_{25}\ot_{35})(\det{\bar{L}})^3$ \\
 \hline
   \end{tabular}
   \caption{The values $\Delta y$ depending on the type of $G$.}
 \label{tab:dy}
 \end{table}
 
\renewcommand{\arraystretch}{1.2}
\begin{table}
   \centering
   \begin{tabular}{cll}
       \hline \hline
       type of $G$ & upper bound for $\Delta x$ & upper bound for $\Delta y$  \\
     \hline \hline
     $3$ &  $\det\bar L$ & $\det \bar L$ \\
     $4$ &  $2\det\bar L$ & $2n(\det\bar L)^2$ \\ 
    $5A$ & $4n^2 (\det\bar L)^3$ & $2n (\det\bar L)^2$\\
    $5B$ & $2 \det\bar L$ & $16 n^2(\det{\bar{L}})^3$ \\
 \hline
   \end{tabular}
   \caption{Upper bounds for  $\Delta{x}$ and $\Delta y$ depending on the type of $G$.}
 \label{tab:upperbounds}
 \end{table}

We finish the analysis of the necessary grid size by calculating the 
size of the $z$-coordinates. 
\begin{lem}\label{lem:zsize}
Let $G(\p)$ be an integral 2d embedding of a graph with $n$ vertices
with equilibrium stress $\omega$ and let the
stress on all interior edges be $1$. The difference between two $x$-coordinates is at most
$\Delta x$ and the difference between two $y$-coordinates is at most $\Delta y$. Then we have an integral lifting with
\[0\leq z_i < 2n\Delta x \Delta y \]
for all $z$-coordinates $z_i$.
\end{lem}
\begin{proof}
Due to Lemma~\ref{lem:scaling1} we know that there exists an integral 
lifting for the setting described in the lemma.  We place an interior 
face $f_1$ in the $xy$-plane and compute the lifting by using the stresses on
the interior edges. 
Notice that all $z$-coordinates
are non-positive in this lifting. Thus it suffices to compute the smallest
$z$-coordinate. The claimed lifting is then obtained by translating the polytope such that the smallest
$z$-coordinate becomes 0.

We choose as face $f_1$ a face that shares an edge with
the outer face $f_0$. Furthermore we assume that the boundary point
farthest away from the line that contains $f_1\cap f_0$ is located in
the origin (let this point be $\p_1$). 
This is no restriction since a translation of the 
embedding does not interfere with the lifting. The lifted 
polytope lies below the $xy$-plane $H_{1}$ and above  $H_0$.
We notice that the smallest $z$-coordinate of $H_{0}$ (and hence the smallest $z$-coordinate of the embedding) is realized at $\p_1$.

Let $f_k$ be an interior face that contains $\p_1$. The
$z$-coordinate of $\p_1$ is given as 
\[z_1=\langle \mathbf{a_k},\p_1 \rangle + d_k = d_k.\]

The variable $d_k$ can be computed with help of equation
$\eqref{equ:liftingit2}$. Let $\mathcal{C}$ be a set of interior edges 
that are crossed by ``walking'' from $f_1$ to $f_k$. Due to Euler's formula $G$ 
has at most $2n-4$ faces. No face is entered twice and thus every face 
contributes at most one edge (the edge where the ``walk'' leaves the face) to the set $\mathcal{C}$.
This implies that $\mathcal{C}$ includes at most $2n-3$ edges. We ignore the orientation of the edges
at this place since it does not matter for bounding $d_k$. We 
deduce
\begin{align*}
-d_k & \leq  \sum_{(i,j)\in\mathcal{C}}|\langle \p_i,\p_j^\bot
\rangle|
 <  2n  \max\{\,\lvert\langle \p_i,\p_j^\bot\rangle\rvert \colon 1\leq i,j \leq n\}.
\end{align*}
For two points $\p_i,\p_j$ we have 
\[\langle
\p_i,\p_j^\bot\rangle = x_iy_j-x_jy_i.\]
Thus, $\langle
\p_i,\p_j^\bot\rangle$ equals
two times the negative area
of the triangle spanned by $\p_i$, $\p_j$, and the origin (which coincides with $\p_1$). 
This triangle is contained
inside the embedded outer face $f_0$ and also inside a rectangle with edge lengths
$\Delta x$ and $\Delta y$. A rectangle has at least twice the area of
an inscribed triangle. To see this, observe that an inscribed triangle
with the largest area must have one of the rectangle edges as base and
the other as height. Thus  $|\langle
\p_i,\p_j^\bot\rangle|\leq\Delta x\Delta y$ and 
the smallest $z$-coordinate is larger than $-2n\Delta x \Delta
y$.
\end{proof}

By applying Lemma~\ref{lem:zsize}, we compute the bounds for
the $z$-coordina\-tes, using the values of $\Delta x$ and $\Delta y$
listed in
Table~\ref{tab:upperbounds}. We conclude with the main
theorems.
\begin{thm}
Every $3$-polytope with $n$ vertices whose graph contains at least a triangle can
be realized on an integer grid with 
  \[\setlength{\arraycolsep}{0pt}
\begin{array}{rcl}
    0  \leq{}& x_i,y_i&{} <   5.\bar{3}^n, \\
    0 \leq{}& z_i&{} <   2 n \cdot 28.\bar{4}^n. \\
  \end{array}\]
\end{thm}
\begin{thm}\label{thm:4gon}
Every $3$-polytope with $n$ vertices 
whose graph contains at least one quadrilateral face can
be realized on an integer grid with 
  \begin{align*}
   0  \leq  x_i & <  2 \cdot 3{.}530^n, \\
   0  \leq  y_i & <  2 n \cdot 12{.}46	1^n, \\
   0  \leq  z_i & <  8 n^2 \cdot 43{.}987^n. 
  \end{align*}
\end{thm}
For the most general theorem we have to combine the analysis for the
cases $5A$ and $5B$. We rotate the embedding if $G$ is of type $5B$ by
exchanging its $x$ and $y$-coordinates to obtain a better bound. The
largest $z$-coordinate is given by
 $\max\{16 n^4 \cdot 187{.}128^n,64n^3 \cdot 65{.}722^n\}
= 16 n^4 \cdot 187.128^n$,
since $n> 4$ if $G$ contains a
pentagon.
Thus the largest bound on the $z$-coordinate arises from case~$5A$.
\begin{thm}\label{thm:mainresult}
Every 3-polytope with $n$ vertices can be realized on an integer grid with
\begin{align*}
   0  \leq  x_i & <   16 n^2 \cdot 23{.}083^n, \\
   0  \leq  y_i & <  2n  \cdot 8.107^n, \\
   0  \leq  z_i & <  16 n^4 \cdot 187.128^n. 
 \end{align*}
\end{thm}
 We can improve the constant factor of the $z$-coordinate
by a more careful analysis. This can be achieved by placing the face
$f_0$ in the $xy$-plane and then compute the lifting using the
interior edges but also one boundary edge. As mentioned before the
structure of the stresses on the boundary edges is more
complicated. Since the improvement would only be a constant factor we decided
to present the easier analysis with help of Lemma~\ref{lem:zsize}. The
more complicated analysis can be found in~\cite{rrs-epsg-07}.

We see that exponentially large coordinates 
suffice to embed $G$ as $3$-polytope. The 
exponential growth of the size of the coordinates 
is determined by  $(\det \bar L)^5$.
\begin{cor}\label{cor:3d}
Every $3$-polytope  with $n$ vertices 
can be realized with integer coordinates of  size $O(2^{7.55n})$.
\end{cor}

Let us add some remarks on the running time of the embedding algorithm. If we know the substitution stresses the computation of the location of the outer face can be done in constant time. The same is true for the scaling factors. Once we computed the plane embedding, the lifting can be computed face by face, which needs in total $O(n)$ steps.
The computation of  the substitution stresses and of the interior vertices can be done by solving a linear system. Since its underlying structure is planar, we can use nested dissections   based on the planar separator theorem to solve it~\cite{lrt-gnd-79,lt-apst-80}. This implies that a solution can be computed in $O(M(\sqrt{n}))$ time, where $M(n)$ is the upper bound for multiplying two $n\times n$ matrices.  The current record for $M(n)$ is $O(n^{2.325})$ which is due Coppersmith and Winograd~\cite{cw-mmap-90}. Thus the overall running time is given by $O(n^{1.163})$ arithmetic operations.

\section{An Example: the Dodecahedron}
\label{sec:example}
The regular dodecahedron is one of the five Platonic solids. It has
$20$ vertices, $30$ edges and $12$ faces, which are regular
pentagons. Figure~\ref{fig:dodecahedron} shows the graph and a
3-dimensional realization of it.  It is the smallest polytope without triangles and quadrilateral faces, and thus
we have to
apply the more involved cases. 
\begin{figure}[ht]
 \center 
 \begin{tabular}{ccc}
 \raisebox{5mm}{\includegraphics[scale=.34]{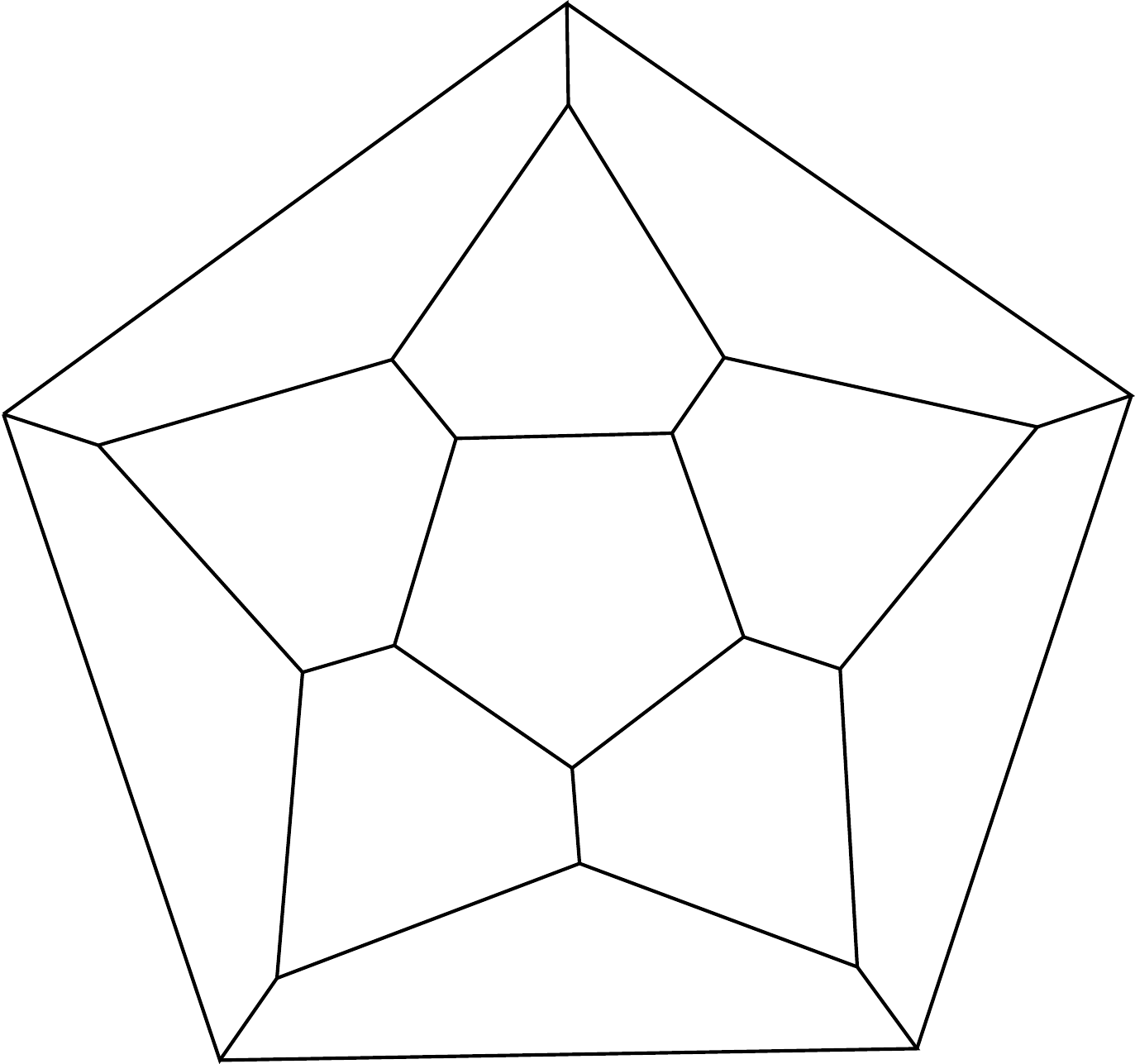}}  &
  \includegraphics[scale=1]{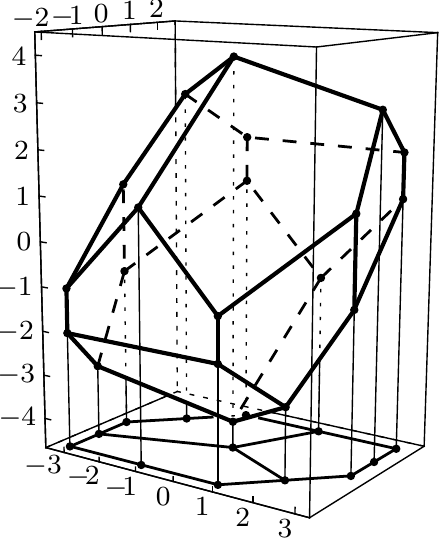} &
 \raisebox{4mm}{\includegraphics[scale=0.8]{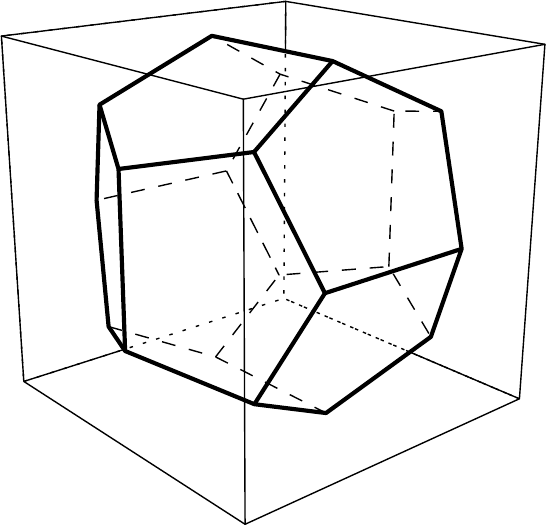}}  \\
     (a) &  (b)&  (c)
\end{tabular}
 \caption{(a) The graph of the dodecahedron and (b,c) two realizations as 
   3-polytope.}
    \label{fig:dodecahedron}
\end{figure}
Since the dodecahedron is symmetric it makes no difference which face we choose as the outer face.
We start the computation with calculating the $\ot$ values. Remember, these values are the off-diagonal entries of the matrix $-(L_{BB}-L_{BI}{\bar L}^{-1}L_{IB})$. We obtain
for all the stresses $\ot_{13},\ot_{14},\ot_{24},\ot_{25}$ and
$\ot_{35}$ the value $36/449$. The fact that all these stresses have the
same value is again due to the symmetry of the dodecahedron.
Since the outer face is a $5$-gon, 
 $G$ is of type $5A$ or $5B$.
Evaluating \eqref{equ:case1} shows that the graph is of type $5A$. With help of the substitution stresses
we compute the coordinates of the boundary vertices. We obtain
\[\mathbf{p}_1=\binom{ 0 }{ 0 },\mathbf{p}_2=\binom{ 1 }{ 0},
 \mathbf{p}_3=\binom{ 1 }{ 1 } ,\mathbf{p}_4=\binom{ 0 }{
   1},\mathbf{p}_5=\binom{ -1/3 }{ 1/2}. \]
We apply Tutte's method to compute the coordinates of the interior points. The result is depicted in 
Figure~\ref{fig:dodecatutte}.
\begin{figure}[ht]
 \center 
   \includegraphics[scale=0.6]{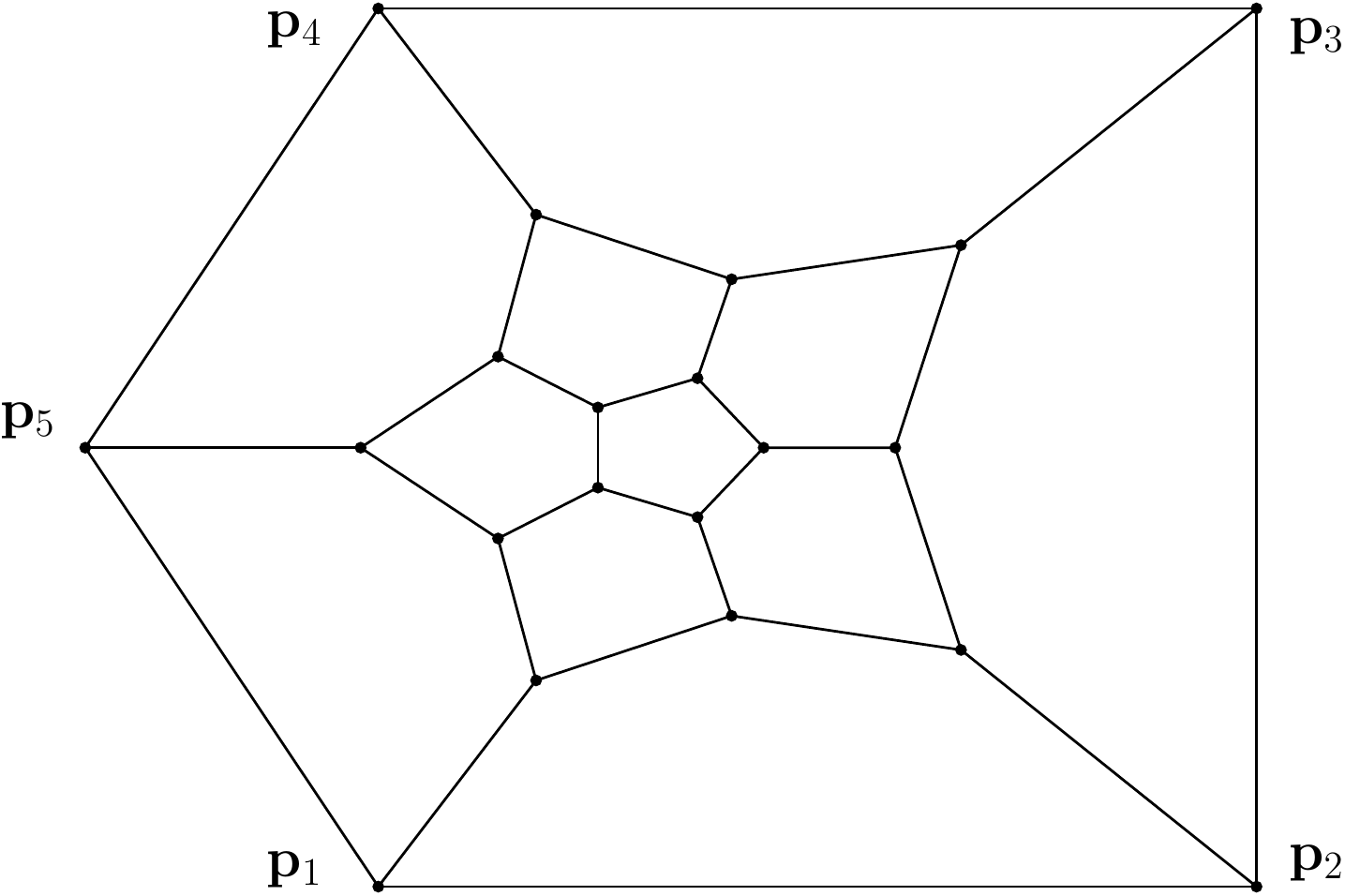} 
    \caption{The barycentric (plane) embedding of the dodecahedron.}
    \label{fig:dodecatutte}
\end{figure}
Next, we scale the 2d-embedding  as
described in Section~\ref{sec:scaling}. We obtain $\det \bar L=403202$. 
This yields the scaling factors
\begin{align*}
  \bar{S}_x =  1\,264\,158\,727\,403\,904, \quad
  \bar{S}_y =  26\,069\,428\,512. 
\end{align*}
We continue with the lifting of the plane embedding to $\mathbb{R}^3$. 
The faces are lifted incrementally as described in
Section~\ref{sec:lifting}. The numeric data of the lifting are listed in \cite{s-lpgri-08}.
Figure~\ref{fig:dodecaresult} shows the computed embedding. We have scaled down the
$z$-coordinates to obtain an illustrative picture.
The highest absolute coordinate is 
\[|z_{3}|= 11\,083\,163\,098\,782\,678\,334\,820\,352 \approx 2^{83.19},\] 
which is smaller than the bound $2^{151}$ from Corollary~\ref{cor:3d}.

\begin{figure}[htb]
 \center 
\includegraphics{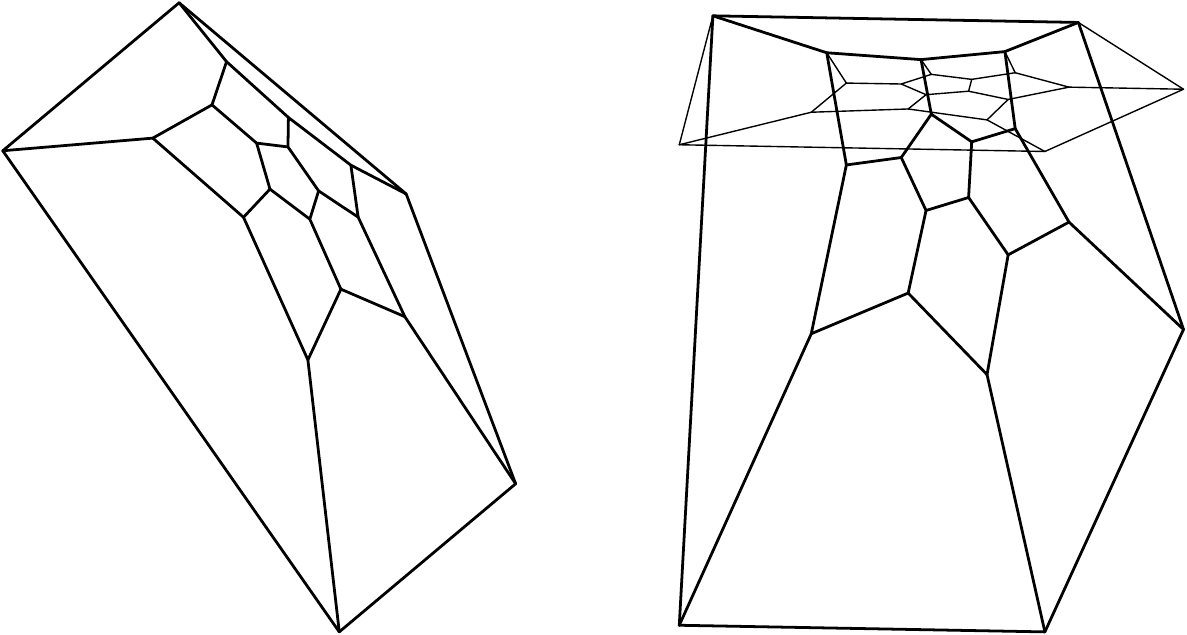} 
    \caption{Two views of the dodecahedron embedded with our algorithm, with scaled $z$-axis.
      The right picture includes also the equilibrium-stressed plane embedding.}
    \label{fig:dodecaresult}
\end{figure}

The computed embedding allows a smaller integer realization. Due to
the fact that the greatest common divisor of the $x$-coordinates is
$938\,499\,426\,432=449^3\times 10365$ and the greatest common divisor
of the $y$-coordinates is $29\,030\,544=449^2\times 144$, scaling down
by these factors yields a smaller integer embedding. We obtain an
integral plane embedding on the grid $[-27,1347]\times[0,898]$.  The
corresponding $z$-coordinates range between 0 and $406\,497$.  This
reduction is due to the fact that all substitution stresses $\ot$ are
equal. Thus one might have replaced them by $\ot\equiv 1$ in the subsequent calculations.

A much smaller grid embedding of the dodecahedron was constructed by hand by
Francisco Santos.
%
%
%
%
 It is centrally symmetric and fits inside a $6\times 4 \times 8$
box, see Figure~\ref{fig:dodecahedron}(b).
It is hard to believe that a smaller realization would be possible.
Another, more symmetric, realization of the dodecahedron is the
pyritohedron (one of the possible crystal shapes of the mineral
pyrite), as pointed out to us by G\'abor G\'evay.  It fits in a
$12\times12\times12$ box, see Figure~\ref{fig:dodecahedron}(c).  It has 8 vertices of the form
$(\pm4,\pm4,\pm4)$, plus 12 vertices, which are the 4 vertices of the form
$(0,\pm3,\pm6)$ and their cyclic rotations of the coordinates.  The
normals of the 12 faces are the vectors $(0,\pm2,\pm1)$ and their
cyclic rotations.

\paragraph{
Acknowledgements.}
We thank a referee for a very thorough reading of the manuscript.

\bibliographystyle{abbrv}
\bibliography{3-polytopes}

\end{document}